\documentclass[11pt]{article}
\usepackage[T1, OT1]{fontenc}
\DeclareTextSymbolDefault{\dh}{T1}
\DeclareTextSymbolDefault{\TH}{T1}
\usepackage{amsmath}
\usepackage{colortbl}
\usepackage{amssymb}
\usepackage{physics}
\usepackage[dvipsnames]{xcolor}
\usepackage{color}
\usepackage{amsthm}
\usepackage{bookmark}

\usepackage{hyperref}
\newtheorem{theorem}{Theorem}

\usepackage{authblk}
\usepackage{graphicx}
\usepackage{physics}
\usepackage{graphicx}
\usepackage{booktabs}
\usepackage{extarrows}
\usepackage{amsmath,amssymb}
\usepackage[a4paper]{geometry}

\usepackage{color,soul}
\usepackage{xspace}

\newcommand{\Vu}{${\sf V_u}$\xspace}
\newcommand{\rc}{{\sf Reg-Center}\xspace}
\newcommand{\vc}{{\sf Vot-Center}\xspace}
\newcommand{\po}{{\sf Poll-Officer}\xspace}
\newcommand{\cc}{{\sf Count-Center}\xspace}
\newcommand{\evot}{{\sf E2E-PQ-EVot}\xspace}

\title{A Post-Quantum Secure End-to-End Verifiable E-Voting Protocol Based on Multivariate Polynomials}

\author[1]{Vikas Srivastava}
\author[2]{Debasish Roy}
\author[3]{Sihem Mesnager}
\author[4,*]{Nibedita Kundu}
\author[5]{Sumit Kumar Debnath}
\author[2]{Sourav Mukhopadhyay}

\affil[1]{Department of Mathematics, IIT Madras; {\tt vikas.math123@gmail.com}}
\affil[2]{Department of Mathematics, IIT Kharagpur; {\tt debasish.roy@maths.iitkgp.ac.in}, {\tt sourav@maths.iitkgp.ac.in}}
\affil[3]{Department of Mathematics, University of Paris VIII, F-93526
	Saint-Denis, University Sorbonne Paris Cité, LAGA, UMR 7539, CNRS, 93430 Villetaneuse, and
	Télécom Paris, 91120 Palaiseau, France; {\tt smesnager@univ-paris8.fr}}
\affil[4]{Krishna Kumar Santosh Kumar Smriti Vidyapith, West Bengal, India  {\tt nknkundu@gmail.com}}
\affil[5]{Department of Mathematics, NIT Jamshedpur; {\tt sdebnath.math@nitjsr.ac.in, sd.iitkgp@gmail.com}}

\begin{document}
\date{}
\maketitle

\begin{abstract}

Voting is a primary democratic activity through which voters select representatives or approve policies. Conventional paper ballot elections have several drawbacks that might compromise the fairness, effectiveness, and accessibility of the voting process. Therefore, there is an increasing need to design safer, effective, and easily accessible alternatives. E-Voting is one such solution that uses digital tools to simplify voting. Existing state-of-the-art designs for secure E-Voting are based on number-theoretic hardness assumptions. These designs are no longer secure due to quantum algorithms such as Shor's algorithm. We present the design and analysis of \textit{first} post-quantum secure end-to-end verifiable E-Voting protocol (namely \evot) based on multivariate polynomials to address this issue. The security of our proposed design depends on the hardness of the MQ problem, which is an NP-hard problem. We present a simple yet efficient design involving only standard cryptographic primitives as building blocks.

\end{abstract}
{\bf Keywords}: Post-quantum cryptography; MPKC;  E-Voting; Decentralized E-Voting

\section{Introduction \label{sec:intro}}

Voting is a fundamental democratic process where citizens of a country or employees of an organization cast ballots to elect representatives or make decisions on policies. It ensures that the government/organization reflects the will of the people. Voting can be conducted through various methods, including traditional paper ballots and electronic systems. Effective voting processes are essential for maintaining the legitimacy of democratic institutions. In addition, it facilitates civic engagement and ensures that diverse voices are heard in decision-making. Most existing state-of-the-art election methods are based on traditional paper ballot elections. As per the report of the highly reputable Pew Research Center: \footnote{\url{https://www.pewresearch.org/short-reads/2020/10/30/from-voter-registration-to-mail-in-\\ballots-how-do-countries-around-the-world-run-their-elections/}}.
\begin{quote}

``Paper ballots are by far the most common form of voting. Votes are cast by manually marking ballots in 209 of the 227 countries and territories for which the ACE Electoral Knowledge Network has data. In addition to paper ballots, electronic voting machines are used in about 10\% of the countries and territories for which data is available. Electronic voting machines are used in some large countries, such as India and the U.S., as well as in smaller ones like Singapore. Voting by internet is used in four countries: Armenia, Canada, Estonia and Switzerland."
\end{quote}

\noindent Although it is historically significant and widely used, paper ballot election has several flaws that can undermine the integrity, efficiency, and accessibility of the electoral process. One of the primary drawbacks of paper ballot elections is its logistical complexity. Organizing a paper-based election requires resources for printing ballots, setting up polling stations, and transporting materials. Moreover, managing and securing these materials to prevent tampering and loss is also a big challenge. Long lines and wait times at polling stations can discourage voters on election day. Security is also a significant issue in paper ballot elections. The physical nature of ballots makes them susceptible to tampering, theft, and destruction. Ballot boxes can be sealed with fraudulent votes, or legitimate ballots can be removed or altered. Ensuring the secure storage and transporting of ballots to counting centers is often vulnerable to breaches.
Moreover, the secrecy of the vote can be compromised if procedures are not strictly followed. Thus, it leads to voter intimidation or coercion. The manual counting of paper ballots is prone to human error. Voters with disabilities may face difficulty accessing polling stations or marking ballots. Similarly, voters in remote or rural areas might need help to reach polling locations. Paper ballots' production, distribution, and disposal consume significant resources and contribute to environmental degradation. To summarize, ensuring complete transparency in the paper-ballot-based electoral process is challenging. Observers and stakeholders must be alert throughout the process to ensure that the election is fair and free from manipulation.

As technology advances, there is a growing need to explore more secure, efficient, and accessible alternatives. E-Voting, specifically online voting is one such solution that simplifies the casting and counting of votes. This method enables voters to participate in elections through secure websites or applications from any location with internet access. In an online voting system, voters typically log in using unique credentials. Once authenticated, voters access an electronic ballot to make their selections. The system securely transmits the cast ballots to a central database. The online voting system may provide a confirmation receipt to the voter or allow them to verify that their vote has been accurately recorded. E-voting has a lot of advantages over traditional paper ballot elections. It significantly reduces the cost of printing, distributing, and managing physical ballots, thereby decreasing the overall expenditure incurred in the election process. Voters can cast their votes from any location with internet access. Thus, it eliminates the need for extensive polling station setups and transportation of materials. E-Voting systems have advanced security features to protect against tampering and fraud. Encryption techniques can be used to secure the transmission and storage of votes, while digital signatures and authentication methods may be utilized to verify voter identities. E-Voting increases accessibility by providing convenient voting options for individuals with disabilities. In addition, voters in remote or rural areas and expatriates can easily participate without traveling long distances. The flexibility of online voting platforms also allows for extended voting periods, reducing time constraints and making it easier for voters to find a convenient time to cast their ballots. E-Voting systems can enhance transparency through secure and verifiable processes. Voters receive confirmation receipts or can verify their votes through a secure system.

The advent of quantum computing poses a significant threat to existing state-of-the-art E-Voting protocols primarily based on number-theoretic hardness assumptions. Quantum computers can solve number theoretic hard problems such as integer factorization and discrete logarithms exponentially faster than classical computers \cite{shor2002introduction}. This undermines the security of widely used cryptographic algorithms, such as RSA and ECC, which form the backbone of many modern E-Voting protocols. Consequently, there is an urgent need for cryptographic systems that can withstand quantum attacks. Multivariate Public Key Cryptography (MPKC) is one such research direction towards building post-quantum\cite{bernstein2017post} secure systems. Multivariate-based schemes are based on the hardness of solving systems of multivariate polynomial equations (also known as the MQ Problem). MQ-problem is NP-hard\cite{garey1979computers} even for a finite field of size two. Multivariate schemes exhibit compact key sizes and computational efficiency, making them well-suited for E-Voting systems. \\\\

\noindent {\it Related works}. Most of the existing state-of-the-art E-Voting protocols are based on lattice-based cryptography. Chillotti \cite{chillotti2016homomorphic} presented a  post-quantum electronic voting protocol based on LWE fully homomorphic encryption. \cite{chillotti2016homomorphic} provides transparency due to the use of public homomorphic operations. In addition, it eliminated the need for zero-knowledge proofs to validate individual ballots or the correctness of the final election result. Moreover, the protocol in \cite{chillotti2016homomorphic} required only two trustees, in contrast to classical proposals that used threshold decryption via Shamir’s secret sharing. Ronne et al. \cite{ronne2020short} presented an approach for performing the tallying work in the coercion-resistant JCJ voting protocol \cite{juels2005coercion} in linear time using fully homomorphic encryption (FHE). Their proposed enhancement aimed to make \cite{juels2005coercion} quantum-resistant while maintaining its underlying structure. In addition, the authors demonstrated how the removal of invalid votes could be achieved in linear time by FHE, hashing, zero-knowledge proofs of correct decryption, verifiable shuffles, and threshold FHE.
Boyen et al.\cite{boyen2020verifiable} designed and analyzed the first lattice-based verifiable and practical post-quantum mix net with external auditing. The authors made significant contributions by formally proving that their mix net provides a high level of verifiability. Furthermore, the system guarantees accountability by enabling the identification of misbehaving mix servers. A multi-candidate electronic voting scheme employing identity-based fully homomorphic encryption and digital signature algorithm was proposed by Liao et al. \cite{liao2020multi}. A traceable ring signature is a variant of a ring signature that allows a signer to sign a message anonymously labeled with a tag on behalf of a group of users. The identity of the signer is released if two signatures are created with the same tag. Feng et al. \cite{feng2020traceable} proposed an E-Voting mechanism utilizing traceable ring signatures. \cite{farzaliyev2021improved} presents a new design for lattice-based post-quantum mix-nets for E-Voting. This improved architecture enhances the efficiency of zero-knowledge proofs. Kaim et al. \cite{kaim2021post} introduced a post-quantum online voting scheme based on lattice assumptions. It avoids the use of traditional homomorphic primitives and mix-nets. \cite{kaim2021post} utilized a threshold blind signature scheme and lattice-based encryption to construct a robust publicly verifiable E-Voting protocol that can handle complex ballots efficiently.

\subsection{Our Contribution \label{sec:our-contri}}
The major contributions of the paper are listed below.

\begin{itemize}
\item	We proposed the \textit{first} multivariate-based end-to-end verifiable post quantum secure E-Voting protocol (namely \evot). The security of \evot relies on the difficulty of solving a system of multivariate quadratic equations (also known as MQ problem).  Thus, \evot is secure against attacks from quantum computers.

\item \evot is a \textit{generic} protocol. In particular, the design of \evot involves standard cryptographic primitives such as encryption, signature, commitment schemes, and hash functions. Employing these standard primitives simplifies the security and performance analysis of \evot.

\item Our protocol achieves end-to-end (E2E) verifiability by ensuring compliance with the three key properties of verifiability in electronic voting. Firstly, it guarantees that each voter can verify that their intended candidate is correctly captured in the cast vote. In addition, \evot ensures that the cast vote is accurately recorded by the system. Finally, \evot also enables both voters and observers to verify that all recorded votes are tallied correctly. This is facilitated through the use of cryptographic proofs and public bulletin boards, which provide transparency and allow independent verification of the tallying process. Thus, \evot achieves robust E2E verifiability, and hence, ensures the integrity and trustworthiness of the voting process.

\item Anyone can verify the voting results by accessing the key of the \cc. Therefore, complete transparency and fairness is ensured through \evot.

\item The most advanced E2E E-Voting solutions in the current state-of-the-art require a public bulletin board. In practice, bulletin board is difficult to manage.

\item \evot comprises of the five entities: Voter \Vu, Registration Center \rc, Polling Officer \po, Voting Center \vc, and Counting Center \cc. Our proposed design is secure against attacks by malicious \po, \cc, and \vc. In addition, we show that \evot is secure against collusion attacks and attacks by insider and outsider users.

\item \evot is fast and efficient. It can also work on memory-constraint devices because the fundamental operations required in multivariate-based \evot are only modular field multiplications and additions over a field of characteristic two.

\end{itemize}

\subsection{Paper Organization}
The manuscript is organized as follows. We introduce the E-Voting and related works in the post-quantum context in Section \ref{sec:intro}. The significant contributions of the manuscript are summarized in Section \ref{sec:our-contri}. Section \ref{sec:prelims} contains the preliminary information about system building blocks such as multivariate-based signature and commitment protocol. We present the design of \evot in Section \ref{sec:proposed} followed by security analysis in Section \ref{sec:security} and performance analysis in Section \ref{sec:performance}. The paper is concluded in Section \ref{sec:conclusion}.

\section{Preliminaries \label{sec:prelims}}

\subsection{End-to-End (E2E) Verifiable E-Voting}

Verifiability in the E-Voting is typically defined as follows:

\begin{itemize}
\item \textit{Cast-as-intended:} A voter can verify that their intended candidate is correctly captured in the cast
	vote.
\item \textit{Recorded-as-cast:} A voter can verify that their cast vote is correctly recorded by the system.
\item \textit{Tallied-as-recorded:} A voter (or any observer) can verify that all recorded votes are tallied correctly.
\end{itemize}

Voting systems that satisfy the above three properties are called \textbf{End-to-End (E2E)} \cite{hao2016real} verifiable. 

\subsection{System Building Blocks \label{sec:sys-block}}

\noindent {\bf Multivariate Encryption Syntax \cite{ding2009multivariate}}\label{preliminaries:2}
A general MPKC-based encryption scheme contains the following three algorithms. We denote a finite field of order $q$ by $\mathbb{F}_q$.
\begin{itemize}
	\item ${\sf MQE.Kg}$: It generates a secret key $(\mathcal{L},\mathcal{F},\mathcal{T})$ and a public key $\mathcal{P}=\mathcal{L}\circ \mathcal{ F} \circ \mathcal{ T}$, where $\mathcal{ L}:\mathbb{F}_q^m\rightarrow \mathbb{F}_q^m $ and $\mathcal{ T}:\mathbb{F}_q^n\rightarrow \mathbb{F}_q^n$ are two invertible affine maps, and $\mathcal{F}: \mathbb{F}_q^n\rightarrow \mathbb{F}_q^m $ is an easily invertible function, known as ``Central Map". Thereby, $\mathcal{P}$ is a system of $m \in \mathbb{Z}$ number of multivariate polynomials in $n \in \mathbb{Z}$ number of variables.
	
	\item ${\sf MQE.Enc}$: Given a message $\mathbf{x}\in\mathbb{F}_q^m$ and a public key $\mathcal{P} = \mathcal{L}\circ \mathcal{ F }\circ \mathcal{T}$, the encryptor
	derives the ciphertext $\mathbf{y} = \mathcal{P}(\mathbf{x}) \in \mathbb{F}_q^m$.
	
	\item ${\sf MQE.Dec:}$: Utilizing the secret the secret key $(\mathcal{L},\mathcal{F},\mathcal{T})$, the decryptor will decrypt a ciphertext  $\mathbf{y} \in \mathbb{F}_q^m$. For that the decryption  recursively   needs calculates $\mathbf{z} = \mathcal{L}^{-1}(\mathbf{y}) \in \mathbb{F}_q^m$, $\mathbf{w} = \mathcal{F}^{-1}(\mathbf{z}) \in \mathbb{F}_q^n$ and $\mathbf{x} = \mathcal{ T}^{-1}(\mathbf{w}) \in \mathbb{F}_q^n$. Finally, it outputs $\mathbf{x} \in \mathbb{F}_q^n$ as the plaintext corresponding to the ciphertext $\mathbf{y} \in \mathbb{F}_q^m$.
\end{itemize}
Figure~\ref{m-enc} shows the communication flow of a multivariate encryption scheme. We assume that {\sf MQE} is an IND-CCA secure encryption scheme.\\\\

\begin{figure}
	\centering
	\includegraphics[scale=0.18]{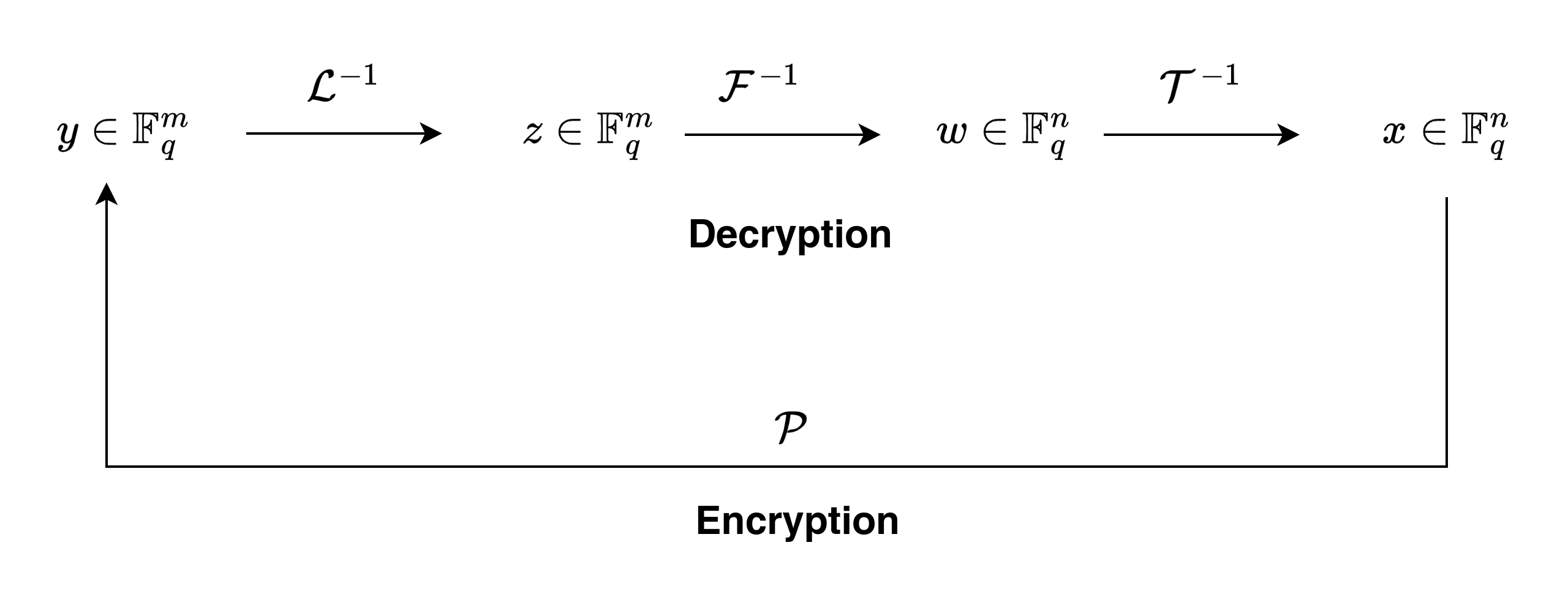}
	\caption{Multivariate Encryption}
	\label{m-enc}
\end{figure}
\noindent {\bf Multivariate Signature Syntax \cite{ding2009multivariate}}\label{preliminaries:4} The following three algorithms are involved in an MPKC based signature scheme.
\begin{itemize}
	\item ${\sf MQS.Kg}$: A central map, $\mathcal{F} : (\mathbb{F}_q)^n \rightarrow (\mathbb{F}_q)^m $, which is a quadratic map has to be selected, where we can find its pre-image easily. $m=m(\kappa)$ multivariate quadratic polynomial in $n=n(\kappa)$ variables are use to construct the structure of $\mathcal F$  with the security parameter $\kappa$. We use invertible affine transformations $\mathcal S : (\mathbb{F}_q)^m \rightarrow (\mathbb{F}_q)^m$ and $\mathcal T : (\mathbb{F}_q)^n \rightarrow (\mathbb{F}_q)^n$ to mask the design of $\mathcal F$ in the public key. $(\mathcal S,\mathcal F,\mathcal T)$ and $\mathcal{P}=\mathcal S\circ \mathcal F \circ \mathcal T$ are chosen as the signing key and the verification key, respectively.
	
	\item ${\sf  MQS.Sign}$: The signer recursively computes $\mathbf{y} = \mathcal S^{-1}(\mathbf{x}) \in (\mathbb{F}_q)^m$, $\mathbf{z} =\mathcal F^{-1}(\mathbf{y}) \in (\mathbb{F}_q)^n$ and $\mathbf{w} = \mathcal T^{-1}(\mathbf{z}) \in (\mathbb{F}_q)^n$ for signing a text $\mathbf{x}\in (\mathbb{F}_q)^m$. Finally, $\mathbf{w}$ is provided as $\mathbf{x}$'s signature. Note that $\mathcal F^{-1}(\mathbf{y})$ signifies one (among perhaps many) pre-images of $\mathbf{y}$ under~$\mathcal F$.
	
	\item ${\sf MQS.Ver}$: If equality $\mathbf{x} \stackrel{?}{=} \mathcal{P}(\mathbf{w})$ holds, a message-signature pair $(\mathbf{x},\mathbf{w})$ is accepted by the verifier, otherwise rejects.
\end{itemize}
Figure~\ref{m-sig} shows the communication flow of a multivariate signature scheme. We assume that {\sf MQS} is existentially unforgeable against chosen message attack.\\
\begin{figure}
	\centering
	\includegraphics[scale=0.15]{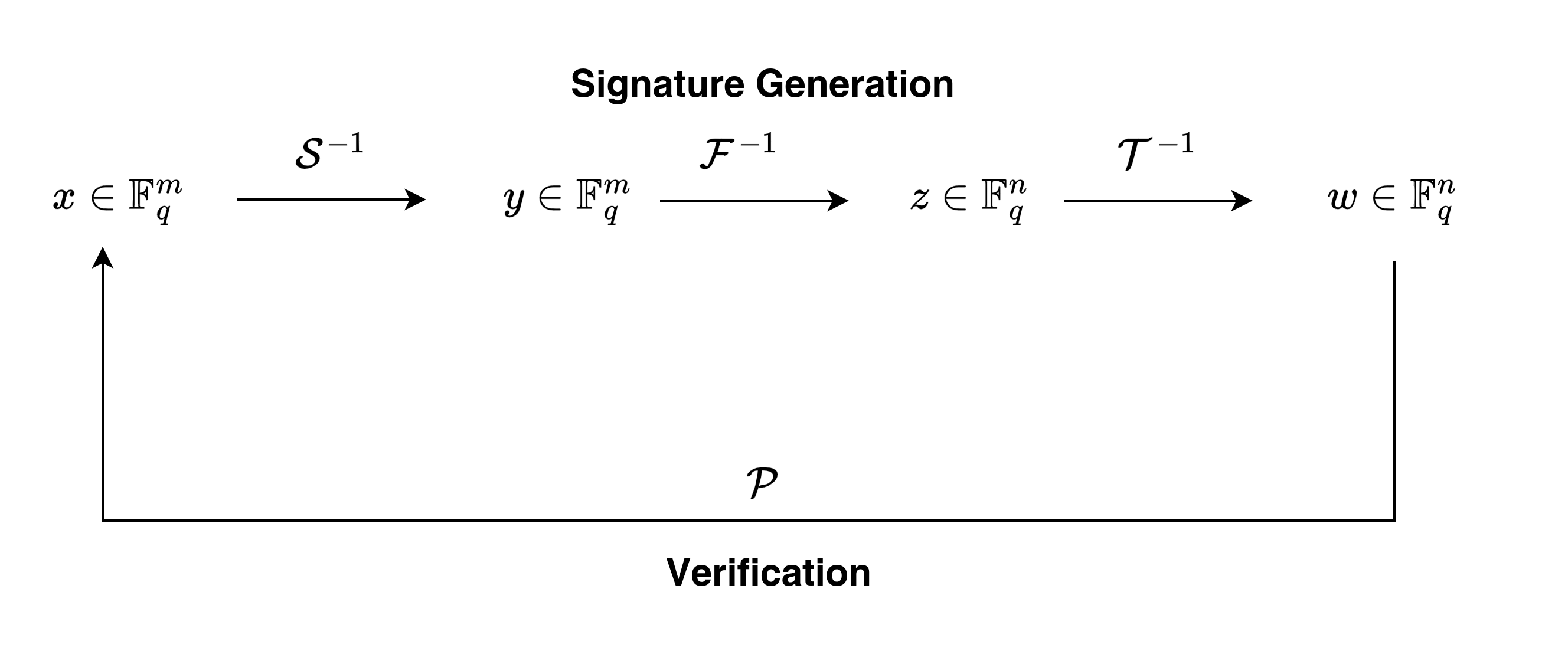}
	\caption{Multivariate Signature}
	\label{m-sig}
\end{figure}

\noindent {\bf A Statistically Binding and Computationally Hiding Commitment}: A commitment scheme ${\sf Commit} = ( {\sf Comm, Open})$ consists of two PPT algorithms which are defined below.
\begin{enumerate}
	\item  $( c , r ) \leftarrow {\sf Comm}(m)$: Algorithm {\sf  Comm} takes a message $m$ as input, and outputs a commitment-opening pair $(c , r)$.
	\item $1/0 \leftarrow {\sf Open}(m, c , r)$ : Algorithm {\sf Open} takes a message $m$, a commitment $c$ and an opening $r$ as input, and outputs $1$ or $0$ to indicate the validity or invalidity of the opening.
	
\end{enumerate}

The primary security features needed in {\sf Commit} are the following. A cheating receiver cannot learn $m$ from the commitment protocol (hiding property). A cheating committer cannot change his mind about $m$, and the verifier can
check in the opening that the value opened was what the committer had in mind originally (binding property). Each of the properties above can be satisfied either statistically or relative to a complexity assumption. For our generic construction of \evot, we require a commitment scheme {\sf Commit} that satisfies statistical binding and computational hiding.

\section{Proposed E2E Post-Quantum E-Voting Protocol {\sf PQ-EVot} \label{sec:proposed}}
In this section, we present the detailed design of post-quantum secure E-Voting protocol \evot.

\subsection{System Architecture}

We first present the system architecture of {\sf PQ-EVot}. Refer to Figure \ref{fig-total} for an illustration. The system comprises of the five entities: Voter \Vu, Registration Center \rc, Polling Officer \po, Voting Center \vc, and Counting Center \cc. Here, the voter \Vu is the entity that intends to cast a vote in the election. \Vu represents a member of the organization or citizen of a country. The proof that \Vu is a genuine voter means that \Vu should possess a citizenship card, social security card, or some other digital identification certificate containing verifiable data. \Vu, with a valid identity, registers his intention to vote with the \rc. \Vu sends his identity information and pseudonym to \rc. In the following, \rc checks whether the \Vu's identity is valid from the electoral roll databank. If the voter's identity is valid, \rc signs the pseudonym and produces a ticket, which the voter can use for voting. \rc maintains a database of voters' identities, the tickets they were issued, and their pseudonyms. If another person approaches the \rc with the same identity, he can verify from the database and refuse to issue a ticket to the voter. \Vu submits the ticket at the current time to the \po, which authorizes \Vu to vote. The \po verifies the ticket and prepares a ballot paper with the names of candidates and the \po's signature so that the \Vu can vote. The \po casts his choice on the ballot and sends the encrypted ballot paper and a ticket to the \po. The \po collects and sends all such cast ballot papers to the \po. In addition, \po collects the receipt for each ballot paper submitted to the \vc.  After obtaining the ballots, the \vc authenticates the credibility of the cast ballots and the ticket. After verification, \vc publishes the encrypted cast ballot. The voters can check whether their vote has been counted from the bulletin board. The \cc collects and decrypts the encrypted ballot from the Bulletin Board. The main assumption here is that the Counting Center is trusted. CC does this process for each ballot, computes the tally, and publishes the final result on the Bulletin Board. After the result's publication, \cc will publish his secret key on the bulletin board, using which anyone can verify the election results.

\begin{figure}
	\centering
	\includegraphics[scale=0.08]{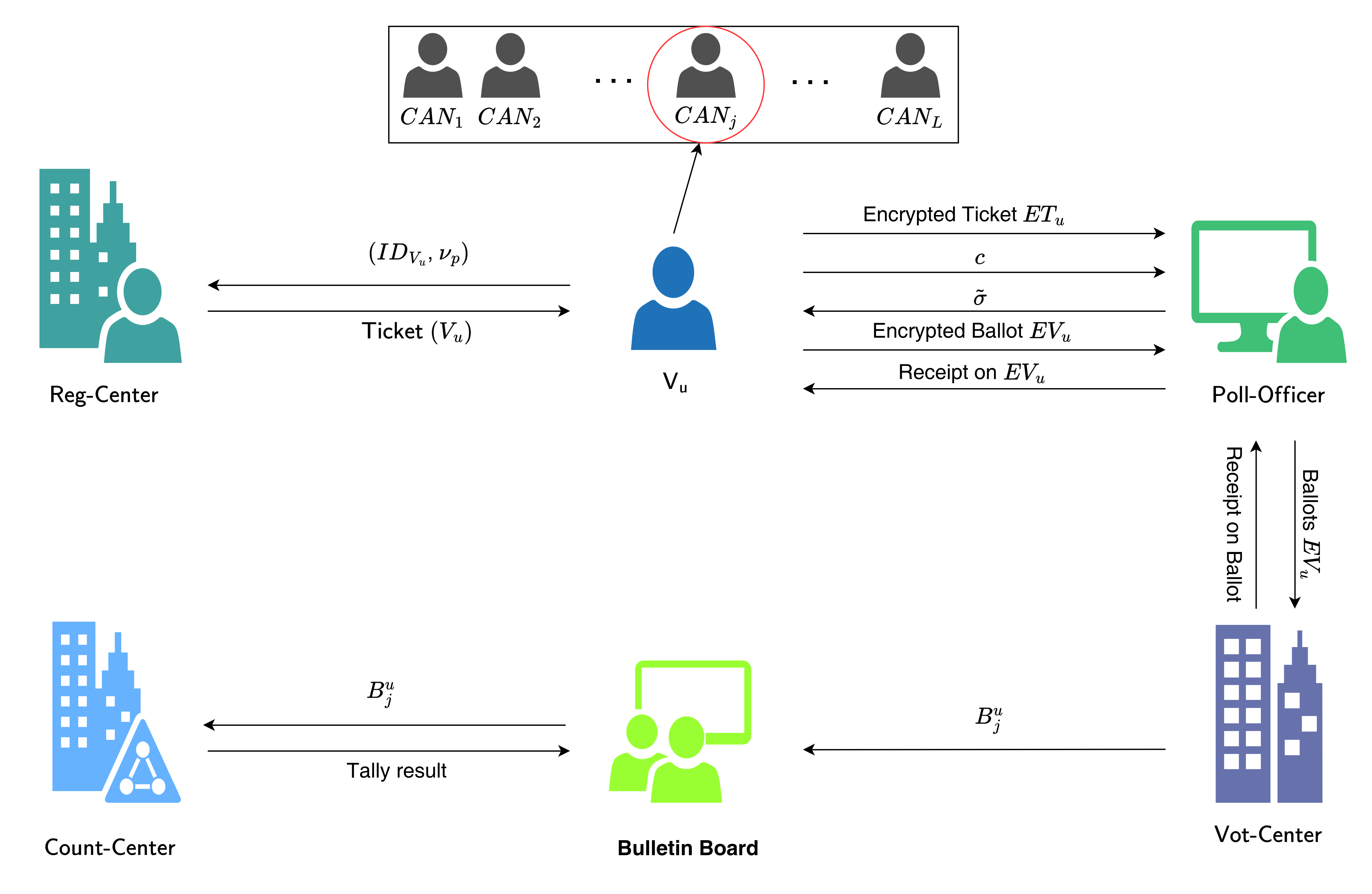}
	\caption{System architecture and communication flow of \evot.}
	\label{fig-total}
\end{figure}

\subsection{Threat/Attack Model}
The security of our designed protocol \evot is based on the fact that solving a random system of multivariate quadratic equations is NP-hard. In particular, the MQ problem is defined as follows.\\
\noindent {\bf{The MQ Problem \cite{ding2009multivariate}:}} Let $\{p_1(x_1,\ldots,x_{\eta}),\ldots,p_{\lambda}(x_1,\ldots,x_{\eta})\}$ be a system of second degree multivariate polynomials over $\mathbb{F}_q$. We have to find a solution $\mathbf{x}=(x_1,\ldots,x_{\eta})$ of the system of equations $p_1(\mathbf{x})=\ldots=p_{\lambda}(\mathbf{x})=0$.
\indent It is worth underlining that the hardness of the MQ problem relies on the number of the variables (i.e. $\eta$) and polynomials (i.e. $\lambda$). Even for polynomials of degree $2$ over $\mathbb{F}_2$, the MQ problem has been demonstrated to be NP-hard~\cite{garey1979computers}. 

\noindent Here, we assume that \rc is a trusted authority. \po, \cc, and \vc may act maliciously.  Therefore, the attack models is stated as following.
\begin{itemize}
	\item The Voter \Vu is not trusted, and the voter is subject to coercion. 
	\item Registration Center \rc is trusted. 
	\item \po, \cc, \vc is not trusted and may try to sabotage the election. 
	\item \po, \cc, \vc may collude together and try to gain some information from the protocol.
	\item The channel of communication is open unless specified otherwise.
\end{itemize}
The security features for the \evot system are designed to ensure its reliability and integrity. We discuss the desired properties for \evot. Firstly, \textit{eligibility} requires that only eligible voters are allowed to cast ballots, ensuring that the voting process is restricted to those who are authorized to participate. \textit{Non-reusability} ensures that voters can cast only one ballot. It prevents any individual from voting multiple times. \textit{Soundness} guarantees that no unauthorized person can alter another voter's ballot. \textit{Completeness} provides that all voters can verify their votes are included in the total count, ensuring transparency and accuracy in the voting results. \textit{Verifiability} protects against fraudulent activities, ensuring that the final voting results are legitimate and tamper-proof. \textit{Fairness} ensures that no one can access data regarding the tally outcome before the official tally phase, maintaining the integrity of the election process. \textit{Anonymity} safeguards voter privacy by ensuring that no one can link a vote to the voter who cast it. We assume the following model proposed by Quang et al. \cite{50AttackerModel} for our proposed design.

\begin{itemize}
	\item Any communication between protocol actors via a public channel could be intercepted by an attacker.
	\item The intercepted message can be altered, removed, sent again, or redirected by an attacker.
	\item Since the attacker is aware of the description of the protocol, it is considered public.
\end{itemize}
\subsection{Design of {\sf PQ-EVot}}

In this section, we propose the design of our post-quantum secure E-Voting system {\sf PQ-EVot}

\begin{description}
	\item {${\sf Preparation\text{ }Phase} (1^{\kappa}$)}: In the preparation phase, \rc, \po, \vc, and \cc each generates their public key and private key pair by utilizing a secure multivariate-based encryption scheme {\sf MQE} = ({\sf MQE.Kg, MQE.Enc, MQE.Dec}). Let the private key-public key pair of \rc, \po, \cc, and \vc be denoted by $(sk_{E_{{RC}}}, pk_{E_{{RC}}})$, $(sk_{E_{{PO}}}, pk_{E_{{PO}}})$, $(sk_{E_{{CC}}}, pk_{E_{{CC}}})$, and \\$(sk_{E_{{VC}}}, pk_{E_{{VC}}})$ respectively. 
	A random system of multivariate polynomial $\mathcal{P}:\mathbb{F}_q^n\rightarrow \mathbb{F}_q^m$ is also generated as a part of the system parameter. Additionally, \rc, \po, \vc, and \cc employ a secure multivariate signature protocol {\sf MQS}= ({\sf MQS.Kg, MQS.Sign, MQS.Ver}) to generate the signature public key and signature private key. Let the signature private key-public key pair of \rc, \po, \cc, and \vc be denoted by $(sk_{S_{{RC}}}, pk_{S_{{RC}}})$, $(sk_{S_{{PO}}}, pk_{S_{{PO}}})$, $(sk_{S_{{CC}}}, pk_{S_{{CC}}})$, and $(sk_{S_{{VC}}}, pk_{S_{{VC}}})$ respectively. The final system public parameter ${\sf pp}$ is set as $${\sf pp=}( pk_{S_{{RC}}}, pk_{S_{{PO}}}, pk_{S_{{CC}}}, pk_{S_{{VC}}}, pk_{E_{{RC}}}, pk_{E_{{PO}}},pk_{E_{{CC}}}, pk_{E_{{VC}}}, \mathcal{P})$$
	\item ${\sf Registration\text{ }Phase}$: To join the system, a voter \Vu has to register to \rc first. This ensures that only authorized voters can participate in voting. The registration phase is executed as follows:
	\begin{itemize}
		\item Voter selects $a\in \mathbb{F}_q^n$ and computes $v_p=\mathcal{P}({\sf ID}_{\sf V_u}\Vert a) \in \mathbb{F}_q^m$ as the pseudonym. 
		\item The voter \Vu sends $({\sf ID}_{V_u}, v_{p} )$ to the \rc. 
		\item Given the identity information ${\sf ID}_{V_u}$, \rc checks the qualification of \Vu for voting. If the information is correct and \Vu is valid for voting, \rc proceeds to verify the identity to ensure that the forger cannot use some other person's identity to vote on his/her behalf.
		\item \rc signs $v_p$ by using any secure MQ-based signature algorithm ${\sf MQS.Sign}\\(v_p, sk_{S_{RC}})$ to obtain the signature $s_p$. 
		\item \rc returns the certification $ {\sf Ticket(V_u)}=(v_P,s_P)$ to the voter \Vu. \rc publishes all pseudonyms that have obtained certifications. The communication flow of ${\sf Registration\text{ }Phase}$ is illustrated in Figure \ref{reg-phase}.
		
	\end{itemize}
	\begin{figure}
		\centering
		\includegraphics[scale=0.18]{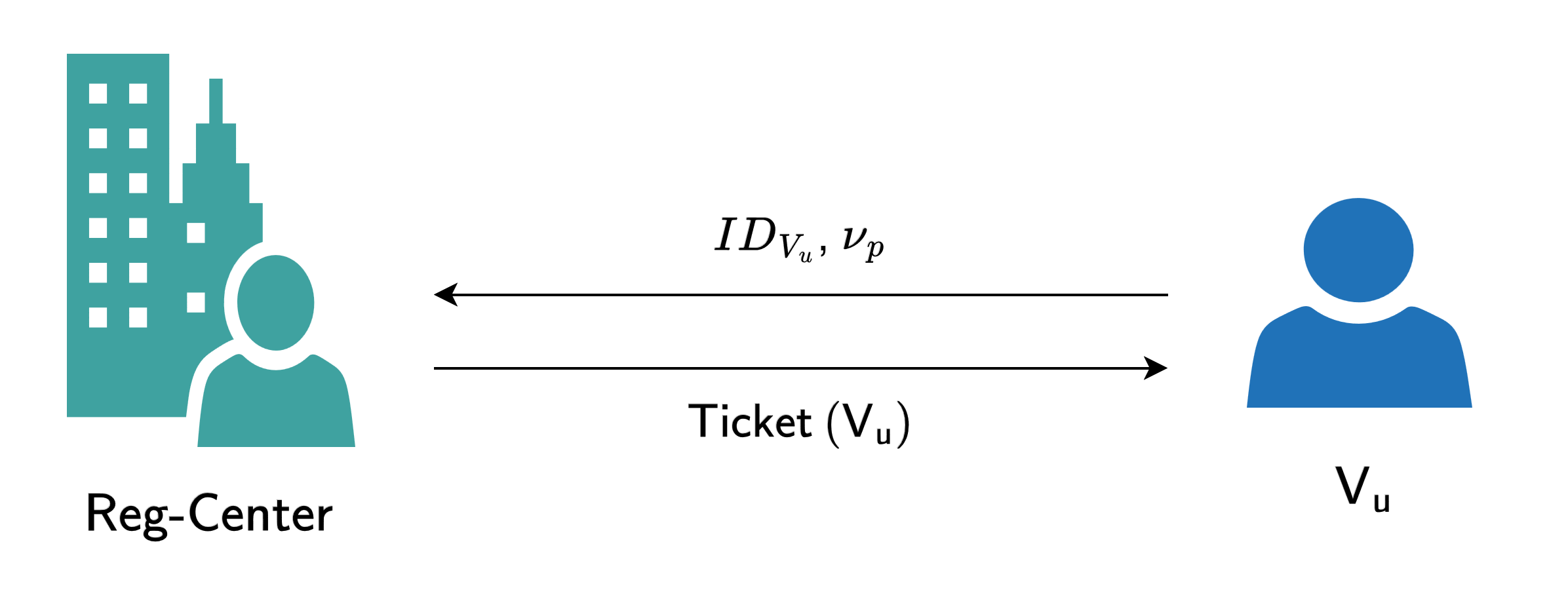}
		\caption{Communication flow in the ${\sf Registration\text{ }Phase}$ of \evot }
		\label{reg-phase}
	\end{figure}

	\item ${\sf Voting\text{ }Phase}$: In voting phase, assume that a legitimate voter \Vu would like to get a signature on a candidate $CAN_j \in \{CAN_1, CAN_2,\ldots, CAN_L \}$ which is obliviously signed by \po, then \Vu executes the following protocol with \po. We use the oblivious signature technique of \cite{zhou2022generic} as a building block in this phase:
	\begin{itemize}
		\item \Vu sends $ET_u={\sf MQE.Enc}({\sf Ticket(V_u)}||tm_u, pk_{E_{PO}})$ to the \po. Here $tm_u$ represents the current time. 
		\item \po verifies whether $tm_u$ and ${\sf Ticket(V_u)}$ are valid, and ${\sf Ticket(V_u)}$ is not reused. If it is valid without reusing, \po stores ${\sf Ticket(V_u)}$ in the database.
		\item  \Vu computes $c, r\leftarrow {\sf Comm}(CAN_j)$ and sends $c$ to \po. Here {\sf Comm} is a secure commitment protocol. From $i=1,\ldots, L$, \po computes $$\sigma_i\leftarrow {\sf MQS.Sign}(sk_{S_{PO}}, CAN_i\mathbin\Vert c)$$
		\item \po sets $\tilde{\sigma}=(\sigma_1, \ldots, \sigma_L)$ and sends $\tilde{\sigma}$ to $V_u$.
		\item From $i=1, \ldots,L$, if $${\sf MQS.Ver}(CAN_i\mathbin \Vert c, \sigma_i, pk_{S_{PO}})=0, \text{ returns $\perp$}$$ \Vu obtains the signature $\Sigma = (\sigma_j,c,r)$ on $CAN_j$.	
		\item \Vu computes $B_j^u={\sf MQE.Enc}(\Sigma||CAN_j, pk_{E_{CC}})$, and $EV_u={\sf MQE.Enc}({\sf Ticket(V_u)}\\||B_j^u, pk_{E_{VC}})$ and sends $EV_u$ to \po. \po provides a receipt on $EV_u$ (say $\sigma_{EV_u, PO}$) to the voter \Vu. 
		\item \po collects all the cast ballots and sends them to \vc. While submitting the cast ballots to \vc, \po obtains a receipt from \vc for each of the submitted ballots. The communication flow of ${\sf Voting\text{ }Phase}$ is illustrated in Figure \ref{vot-phase}.
	\end{itemize}
	
		\begin{figure}
		\centering
		\includegraphics[scale=0.115]{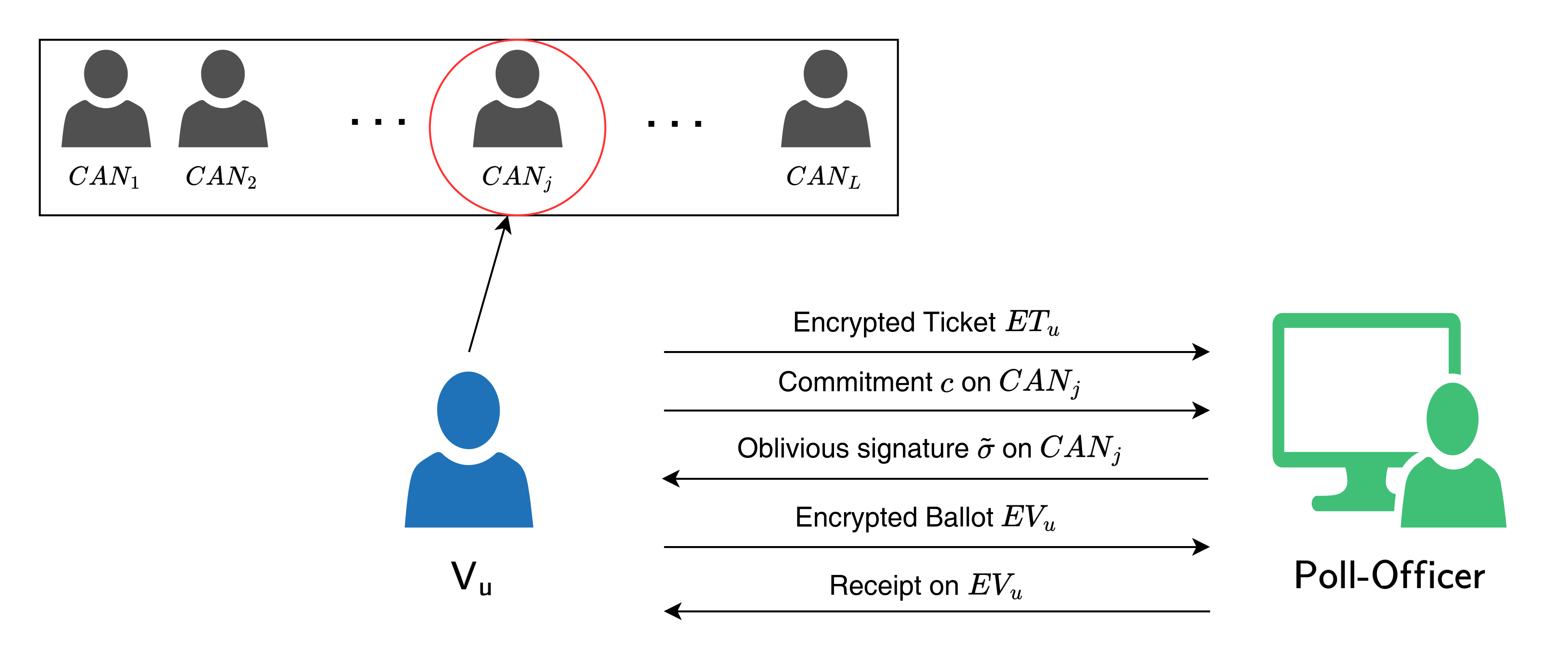}
		\caption{Communication flow in the ${\sf Voting\text{ }Phase}$ of \evot }
		\label{vot-phase}
	\end{figure}

	\item ${\sf Verification\text{ }Phase}$:  In the verification phase, \vc first verifies the cast ballots and, in the following, publishes the result on the bulletin board.
	\begin{itemize}
		
		\item \vc computes ${\sf MQE.Dec}(EV_u)={\sf Ticket(V_u)}||B_j^u$, verifies ${\sf Ticket(V_u)}$, stores $({\sf Ticket(V_u)}||B_j^u)$ in database, and writes $B_j^u$ on Bulletin board for each of the cast ballots. Note that $B_j^u$ is public to all.
		\item \Vu can check whether his/her votes $B_j^u$ are displayed on the Bulletin Board.
	\end{itemize}
	
	\item ${\sf Tally\text{ }Phase}$: This phase describes the tally procedure computed by the \cc. 
	\begin{itemize}
		\item \cc obtains $B_j^u $ from Bulletin Board.
		\item \cc computes ${\sf MQE.Dec}(B_j^u, sk_{E_{CC}})$ to obtain $\Sigma||CAN_j$. In the following, \cc checks the validity of the signature $\Sigma$ for $CAN_j$. If the verification succeeds, then add $1$ in front of $CAN_j$ in the bulletin board
		\item \cc repeats the process of each of the $B_j^u$ and publishes the final tally result on the Bulletin Board.
		\item  At the end, \cc  makes his secret key (decryption key) $ sk_{E_{CC}}$ public. Note that anyone can verify the result by accessing the secret key.  
		
	\end{itemize}
	
\end{description}

\section{Security Analysis \label{sec:security}}

\begin{theorem}
\evot is secure from the attacks by a malicious \po.
\end{theorem}
\begin{proof}
	A malicious \po cannot learn which candidate out of $\{CAN_1, \ldots, CAN_L\}$ is chosen by the voter \Vu for casting the ballot. Assume towards a contradiction that there exists a malicious \po that can learn about which candidate was chosen by the voter. We show that we can construct an adversary, say $\mathcal{C}$, against the computational hiding property of {\sf Commit}.
	\begin{itemize}
		\item Firstly, $\mathcal{C}$ invokes key generation algorithm of underlying signature protocol {\sf MQS} to generate the secret key and public key (say $ (sk, pk)$). $\mathcal{C}$ sends $( sk, pk)$ to \po, and receives a set of messages $\{CAN_1, CAN_2, \ldots, CAN_L\}$ from \po.
		\item Then $\mathcal{C}$ samples an index $j \leftarrow [L]$ uniformly, sets
		$CAN_0' := CAN_j$ and samples a message $CAN_1'$ uniformly from the message space.
		\item $\mathcal{C}$ sends $(CAN_0', CAN_1')$ to the	challenger, and receives a challenge ${ c^*}$, where ${ c^*}$ is either a commitment of $CAN_0$ or a commitment of $CAN_1$. $\mathcal{C}$ aims to guess the case.
		\item Finally, $\mathcal{C}$ returns $\delta := { c^*}$ to \po, and receives $j^*$
		from \po as the guessing of $j$. $\mathcal{C}$  returns $1$ to its own challenger if and only if $j^* = j$.
	\end{itemize}

\noindent Now we analyze $\mathcal{C}$'s advantage against the computational hiding property of {\sf Commit}. In the case that $c^*$ is a commitment of $CAN_0'$, $\delta (=c^*)$ is a commitment of $CAN_j (=CAN_0')$. Hence, \po output $j^*$ equals $j$ with probability $1/L\pm negl(\lambda)$ for some non-negligible function $negl$, and consequently $\mathcal{C}$ returns $1$ to its own challenger with the same probability $1/L+ negl(\lambda)$ as well. If $c^*$ is a commitment of $CAN_1'$, $\delta (= c^*)$ is a commitment of $CAN_1'$, a uniformly chosen message independent of $j$. Therefore, $j$ is completely hidden to \po, and \po's output $j^*$ equals $j$ with probability exactly $1/L$. Consequently, $\mathcal{C}$ returns $1$ to its own challenger with the same probability $1/L$. Overall advantage of $\mathcal{C}$ equals to $|1/L\pm negl(\lambda)-1/L|=negl(\lambda)$. This shows that $\mathcal{C}$ breaks the computational hiding property of {\sf Commit} with a non-negligible probability, leading to a contradiction. Therefore, due to the hiding property of the underlying commitment scheme {\sf Commit}, \po cannot gain any information about which $CAN_j$ has \Vu committed to.
\end{proof}

\begin{theorem}
	\evot is secure from attacks by a malicious \cc.
\end{theorem}
\begin{proof}
A malicious \cc may first try to obtain the information about the voter \Vu from $\Sigma|| CAN_j$. Since the signature is verified using the public key of \po, it is not possible for \cc to obtain any information about the identity of the \Vu from $\Sigma|| CAN_j$. Besides, a corrupted \cc may intentionally publish the incorrect tally result on the bulletin board. Note that $B_j^u$ is public to all and accessible from the Bulletin Board. At the end of the  Tally phase, \cc makes the decryption key ($sk_{CC}$) public. Thus, anyone can verify the result due to the access to $sk_{CC}$. Therefore, any malicious activity by \cc can easily be detected. Each individual voter also has a receipt for cast ballots, which they obtained from \po. This receipt can be used by \Vu to claim any fraudulent activity by \cc in the tally phase.
\end{proof}

\begin{theorem}
	\evot is secure from the attacks by a malicious \vc.
\end{theorem}
\begin{proof}
\vc obtains ${\sf Ticket(V_u)}$ and $B_j^u$ after decryption of encrypted cast ballots $EV_u$. From $ {\sf Ticket(V_u)}=(v_P,s_P)$, \vc may obtain the pseudonym $v_p$ of the  voter \Vu. Note that $v_p$ is generated using a random multivariate polynomial system by evaluating $v_p=\mathcal{P}({\sf ID}_{\sf V_u}\Vert a)$, thus, obtaining identity information ${\sf ID}$ from $v_p$ requires the \vc to solve an instance of MQ problem. Thus, \vc cannot determine the relationship between a voter and the vote.  \vc may also try to learn some information about either vote or voter from the cast ballot $B_j^u$. Recall that $B_j^u$ is encrypted with the public key of \cc using a secure encryption scheme {\sf MQE}. Thus, \vc cannot learn about the contents of cast ballot from $B_j^u$. Besides, a malicious \vc may try to cheat by not publishing some of the cast ballots $B_j^u$. \Vu can detect such malicious activity by checking whether his/her cast ballot appears on the bulletin board. \Vu can use the receipt of cast ballot to claim against such fraud. 
\end{proof}

\begin{theorem}
	\evot is secure from attacks by a malicious outside user or an inside voter $\mathcal{V}$.
\end{theorem}

\begin{proof}
In the registration phase, \rc verifies $\mathcal{V}$’s voting qualification by his/her identity ${\sf ID}$. Only legitimate voters can get the ${\sf Ticket}$ from \rc; therefore, only eligible voters could cast the votes. In the voting phase, \po checks if ${\sf Ticket}(V_u)$ is stored in the database; if it is not, \po aborts the process, preventing $\mathcal{V}$ from reusing ${\sf Ticket}(V_u)$. In the tally phase, all valid ballots are published on the Bulletin board and can be verified publicly, preventing invalid ballots by malicious $\mathcal{V}$. From the design of the \evot, it also follows that nobody can forge the signature on behalf of \po. This is because generating a signature on behalf of \po amounts to breaking the unforgeability of the underlying signature and statistically binding property of the underlying commitment scheme.
\end{proof}

\begin{theorem}
	\evot is secure from the collusion attacks.
\end{theorem}
\begin{proof}
	In this attack scenario, we consider the setting where \po, \cc, and \vc collude to corrupt the voting process. Note that the \cc has access to the $\Sigma \Vert CAN_j$. Even with the help of \po and \vc, no information about the voter identity can be obtained from $\Sigma \Vert CAN_j$. Both \po and \vc has access to {\sf Ticket}(\Vu) $=(v_p, s_p)$. But obtaining any information about the identity of a voter from the pseudonym $v_p$ requires solving the MQ problem  $v_p=\mathcal{P}({\sf ID}||a)$. Therefore, even the collusion of the parties mentioned above cannot reveal the information about the voter.   In the registration phase, \rc publishes all pseudonyms that have obtained certifications. Therefore, it is impossible to introduce fake ballots by collusion of \vc, \cc, and \po.
	
\end{proof}

\section{Performance Analysis \label{sec:performance}}
   
Our proposed design \evot employs basic building blocks such as a secure multivariate encryption ({\sf MQE}), a multivariate signature ({\sf MQS}), and a commitment protocol {\sf Commit}. Using standardized cryptographic building blocks such as encryption and signature ensures interoperability with existing systems and protocols. It also helps to integrate the \evot design within different platforms and environments. Since these building blocks are well understood and widely documented, it's easier for independent parties to verify the correctness and security of our proposed protocol. Moreover, using standard primitives simplifies the implementation and maintenance of \evot.\\\\

\noindent {\bf Computation cost:} In the preparation phase of \evot, {\sf MQE.Kg} is invoked four times to generate the public key-private key pairs for \rc, \cc, \po, and \vc. In addition, four executions of underlying {\sf MQS.Kg} algorithm is needed to generate the corresponding signature keys for the abovementioned parties. Moreover, one random system of multivariate polynomial equations is needed to be evaluated in the registration phase. {\sf MQS.Sign} algorithm is performed once by the \rc to generate the {\sf Ticket} for the voter \Vu. In the voting phase, {\sf MQE.Enc} is executed three times during the interaction between \Vu and \po. In addition, the voting phase also requires one instance of {\sf MQE.Dec} and $L+1$ instances each of {\sf MQS.Sign} and {\sf MQS.Ver} algorithm. One instance of {\sf MQE.Dec} algorithm is needed to be performed in the verification phase by the \vc to obtain the {\sf Ticket} and the ballot $B_j^u$ of the voter. \vc executes the {\sf MQS.Ver} algorithm once to verify the validity of the ticket. In the tally phase, we require one execution each of {\sf MQE.Dec} and {\sf MQS.Ver}. Let ${\sf T}_{enc}$ denotes the run-time cost of one instance of {\sf MQE.Enc}, and ${\sf T}_{sig}$ depicts the run-time cost of {\sf MQS.Sig} algorithm. Let ${\sf T}_{kg\mbox{-}sig}$ and ${\sf T}_{kg\mbox{-}enc}$ depict respectively, the run-time cost of key generation algorithm of {\sf MQS} and encryption {\sf MQE}. In addition, let ${\sf T}_{ver}$ and ${\sf T}_{dec}$ represent respectively, the running time of underlying {\sf MQS.Ver} and {\sf MQE.Dec} algorithms. In addition, let ${\sf T}_{eval}$ depicts the time complexity of evaluating a multivariate polynomial system. Table \ref{tab:run-time} summarizes the time complexity of each of the phases.\\\\

\noindent {\bf Storage and communication overhead}: We discuss the storage and communication cost of \evot. Let {\tt |S|} and {\tt |E|} denote respectively, the size of signature and ciphertext produced by the {\sf MQS} and {\sf MQE} respectively. Let {\tt |SK-E|} and {\tt |PK-E|} denote the size of the secret key and public key of {\sf MQE} respectively. Similarly, let {\tt |SK-S|} and {\tt |PK-S|} depict respectively the size of the secret key and public key of {\sf MQS}. Let {\tt |Commit|} denote the commitment size outputted by the underlying {\sf Commit} protocol. We summarize the storage and communication overhead of \evot in Table \ref{tab:storage}.

\begin{table}[]
	\centering
	\begin{tabular}{|l|l|}
		\hline
		
	{\sf Preparation Phase}      & $4{\sf T}_{kg\mbox{-}sig}$ +  $4{\sf T}_{kg\mbox{-}enc}$ \\ \hline
		{\sf Registration Phase}    & ${\sf T}_{eval}$ + ${\sf T}_{sig}$ \\ \hline
		{\sf Voting Phase} & $3{\sf T}_{enc}$ + $(L+1) {\sf T}_{sig}$ + $(L+1) {\sf T}_{ver}$+ ${\sf T}_{dec}$  \\ \hline
		{\sf Verification Phase} & ${\sf T}_{dec}$ + ${\sf T}_{sig}$\\\hline
		{\sf Tally Phase} & ${\sf T}_{dec}$+ ${\sf T}_{ver}$ \\ \hline
	\end{tabular}
	\caption{Run time analysis of \evot}
	\label{tab:run-time}
\end{table}

\begin{table}[]
	\centering
	\begin{tabular}{|l|l|}
		\hline
		Size of ticket      & {\tt |S|}  \\ \hline
		Size of ballot $B_j^u$       & {\tt |E|} \\ \hline
		Size of cast encrypted ballots $EV_u$ & {\tt |E|}  \\ \hline
		Size of public parameter {\sf pp}  & 4{\tt |SK-E|}+ 4{\tt |PK-E|}+ 4{\tt |SK-S|}+ 4{\tt |PK-S|} \\\hline
		Size of vote & {\tt |S|+ {\tt |Commit|}} \\ \hline
	\end{tabular}
	\caption{Storage and communication overhead of \evot}
	\label{tab:storage}
\end{table}

\section{Conclusion \label{sec:conclusion}}

Traditional paper ballot elections have various weaknesses that might compromise the legitimacy, efficiency, and accessibility of the democratic process. E-Voting uses internet-based systems to make voting more accessible and efficient. It addresses the problems in conventional paper ballot elections. In this paper, we designed and analyzed the \textit{first} E-Voting protocol (namely \evot) based on multivariate public key cryptography. The security of our protocol is based on the hardness of the MQ problem. We utilized fundamental cryptographic primitives such as encryption and signature as the building blocks in our design. A thorough security and performance analysis showed that our proposed design is efficient and secure. 

{\subsection*{Data Availability Statement}
\hl{No new data were generated or analyzed in support of this research.}
\subsection*{Funding Details}
This work was supported by CEFIPRA CSRP project number 6701-1.

\end{document}